\documentclass[journal]{IEEEtran}
\usepackage{amsmath,amssymb,physics,tabularx,graphicx,multirow,hhline,algorithm,mathtools,algpseudocode,varwidth,color,soul,amsthm,cite}

\newtheorem{theorem}{Theorem}

\ifCLASSINFOpdf
\else
\fi

\ifCLASSOPTIONcompsoc \usepackage[caption=false,font=normalsize,labelfont=sf,textfont=sf]{subfig}
\else
\usepackage[caption=false,font=footnotesize]{subfig}
\fi

\begin{document}
\algnewcommand{\algorithmicgoto}{\textbf{go to}}%
\algnewcommand{\Goto}[1]{\algorithmicgoto~\ref{#1}}%
\title{A Synchrophasor Data-driven Method\\ for Forced Oscillation Localization \\under Resonance Conditions}

\author{Tong Huang,~\IEEEmembership{Student Member,~IEEE,}
		Nikolaos M. Freris,~\IEEEmembership{Senior Member,~IEEE,}
		P. R. Kumar,~\IEEEmembership{Fellow,~IEEE,}\\
        and Le Xie,~\IEEEmembership{Senior Member,~IEEE}



        }



\maketitle

\begin{abstract}
This paper proposes a data-driven algorithm of locating the source of forced oscillations and suggests the physical interpretation of the method. By leveraging the sparsity of the forced oscillation sources along with the low-rank nature of synchrophasor data, the problem of source localization under resonance conditions is cast as computing the sparse and low-rank components using Robust Principal Component Analysis (RPCA), which can be efficiently solved by the exact Augmented Lagrange Multiplier method. Based on this problem formulation, an efficient and practically implementable algorithm is proposed to pinpoint the forced oscillation source during real-time operation. Furthermore, we provide theoretical insights into the efficacy of the proposed approach by use of physical model-based analysis, in specific by establishing the fact that the rank of the resonance component matrix is at most 2. The effectiveness of the proposed method is validated in the IEEE 68-bus power system and the WECC 179-bus benchmark system.
\end{abstract}


\IEEEpeerreviewmaketitle{}

\section{Introduction} 
\label{sec:introduction}
\IEEEPARstart{P}{hasor} measurement units (PMUs) enhance the transparency of bulk power systems by streaming the fast-sampled and synchronized measurements to system control centers. Such finely-sampled and time-stamped PMU measurements can reveal several aspects of the rich dynamical behavior of the grid, which are invisible to conventional supervisory control and data acquisition (SCADA) systems. 
Among the system dynamical behaviors exposed by PMUs, \emph{forced oscillations} (FOs) have attracted significant attention within the power community. Forced oscillations are driven by periodically exogenous disturbances that are typically injected by malfunctioning power apparatuses such as wind turbines, steam extractor valves of generators, or poorly-tunned control systems \cite{DBLP:journals/corr/GhorbaniparvarZ16,maslennikov2017dissipating,7401141}. Cyclic loads such as cement mills and steel plants, constitute another category of oscillation sources \cite{DBLP:journals/corr/GhorbaniparvarZ16}. The impact of such injected periodical perturbation propagates through transmission lines and results in forced oscillations throughout the grid; some real-world events of forced oscillations since 1966 are reported in \cite{DBLP:journals/corr/GhorbaniparvarZ16}.

The presence of forced oscillations compromises the security and reliability of power systems. For example, forced oscillations may trigger protection relays to trip transmission lines or generators, potentially causing uncontrollable cascading failures and unexpected load shedding \cite{iso_pres}. Moreover, sustained forced oscillations reduce device lifespans by introducing undesirable vibrations and additional wear and tear on power system components; consequently, failure rates and maintenance costs of compromised power apparatuses might increase \cite{iso_pres}. Therefore, timely suppression of forced oscillations is instrumental to system operators.

One effective way of suppressing a forced oscillation is to locate the oscillation's source, a canonical problem that we call \emph{forced oscillation localization}, and to disconnect it from the power grid. A natural attempt to conduct forced oscillation localization could be tracking the largest oscillation over the power grid, under the assumption that measurements near the oscillatory source are expected to exhibit the most severe oscillations, based on engineering intuition. However, counter-intuitive cases may occur when the frequency of the periodical perturbation lies in the vicinity of one of the natural modes of the power system, whence a \emph{resonance phenomenon} is triggered \cite{7041242}. In such cases, PMU measurements exhibiting the most severe oscillations may be geographically far from where the periodical perturbation is injected, posing a significant challenge to system operators in pinpointing the forced oscillation source. It is worth noting that such counter-intuitive cases are more than a mere theoretical concern: one example occurred at the Western Electricity Coordinating Council (WECC) system on Nov. 29, 2005, when a 20-MW forced oscillation initiated by a generation plant at Alberta incurred a tenfold larger oscillation at the California-Oregon Inter-tie line that is 1100 miles away from Alberta \cite{7401141}. Such a severe oscillation amplification significantly compromises the security and reliability of the power grid. Hence, it is imperative to develop a forced oscillation localization method that is effective even in the challenging but highly hazardous cases of resonance\cite{8362302}.

In order to pinpoint the source of forced oscillations, several localization techniques have been developed. 
In \cite{maslennikov2017dissipating}, the authors leverage the oscillation energy flows in power networks to locate the source of sustained oscillations. In this energy-based method, the energy flows can be computed using the preprocessed PMU data, and the power system components generating the oscillation energy are identified as the oscillation sources. In spite of the promising performance of the energy-based method \cite{maslennikov2017dissipating}, the rather stringent assumptions pertaining to knowledge of load characteristics and the grid topology may restrict its usefulness to specific scenarios \cite{8362302}, \cite{bin2017location}.
In \cite{8355784}, the oscillation source is located by comparing the measured current spectrum of system components with one predicted by the effective admittance matrix.
However, the construction of the effective admittance matrix requires accurate knowledge of system parameters that may be unavailable in practice. 
In \cite{8519314}, generator parameters are learned from measurements based on prior knowledge of generator model structures, and, subsequently, the admittance matrix is constructed and used for FO localization. Nevertheless, model structures of generators might not be known beforehand, owing to the unpredictable switching states of power system stabilizers \cite{WECC-standard}. Thus, it is highly desirable to design a FO localization method that does not heavily depend upon availability of the first-principle model and topology information of the power grid.

In this paper, we propose a \emph{purely data-driven} yet \emph{physically interpretable} approach to pinpoint the source of forced oscillations in the challenging resonance case. By leveraging the sparsity of the FO sources and the low-rank nature of high-dimensional synchrophasor data, the problem of forced oscillation localization is formulated as computing the sparse and low-rank components of the measurement matrix using Robust Principal Component Analysis (RPCA) \cite{candes2011robust}. Based on this problem formulation, a algorithm for real-time operation is designed to pinpoint the source of forced oscillations. The main merits of the proposed approach include: 1) it does not require any information on dynamical system model parameters or topology, thus fostering an efficient and easily deployable practical implementation; 2) it can locate the source of forced oscillations with high accuracy, even when resonance phenomena occur; and 3) its efficacy can be interpreted by physical model-based analysis.

The rest of this paper is organized as follows: Section \ref{sec:localization_of_forced_oscillation and its challeges} elaborates on the forced oscillation localization problem and its main challenges; 
in Section \ref{sec:problem_formulation}, the FO localization is formulated as a matrix decomposition problem and a FO localization algorithm is designed; 
Section \ref{sec:theoretical_justification_of_the_proposed_methodology} provides theoretical justification of the efficacy of the algorithm;
Section \ref{sec:case_study} validates the effectiveness of the proposed method in both the IEEE 68-bus power system and the WECC 179-bus power system; Section \ref{sec:conclusion} summarizes the paper and poses future research questions.

\section{Localization of Forced Oscillations and Challenges} 
\label{sec:localization_of_forced_oscillation and its challeges}
\subsection{Mathematical Interpretation} 
\label{sub:mathematical_interpretation_of_forced_oscillation}
	The dynamic behavior of a power system in the vicinity of its operation condition can be represented by a continuous linear time-invariant (LTI) state-space model:
	\begin{subequations}\label{eq:state_space}
		\begin{align}
			\mathbf{\dot{x}}(t)&=A\mathbf{x}(t)+B\mathbf{u}(t),\\ 
			\mathbf{y}(t)&=C\mathbf{x}(t),
		\end{align}
\end{subequations}
where state vector $\mathbf{x} \in \mathbb{R}^{n}$, input vector $\mathbf{u} \in \mathbb{R}^{r}$, and output vector $\mathbf{y} \in \mathbb{R}^{m}$ collect the \emph{deviations} of: state variables, generator/load control setpoints, and measurements, from their respective steady-state values. Accordingly, matrices $A\in \mathbb{R}^{n\times n}$, $B \in \mathbb{R}^{n \times r}$ and $C \in \mathbb{R}^{m\times n}$ are termed as the state matrix, the input matrix and the output matrix, respectively. Denote by $\mathcal{L} = \{\lambda_1, \lambda_2, \ldots, \lambda_n\}$ the set of all eigenvalues of the state matrix $A$. The power system \eqref{eq:state_space} is assumed to be stable, with all eigenvalues $\lambda_i \in \mathbb{C}$ being distinct, i.e., $\Re{\lambda_i}<0$ for all $i\in \{1,2,\hdots,n\}$ and $\lambda_i \ne \lambda_j$ for all $i \ne j$.

We proceed to rigorously define a forced oscillation source as well as source measurements.
Suppose that the $l$-$th$ input $u_l(t)$ in the input vector $\mathbf{u}(t)$ varies periodically due to malfunctioning components (generators/loads) in the grid. In such a case, $u_l(t)$ can be decomposed into $F$ frequency components, viz.,
\begin{equation}
	u_{l}(t)=\sum_{j = 1}^{F}{P_j\sin({\omega_jt + \theta_j})},
	\label{eq:input_sig}
\end{equation}
where $\omega_j\ne 0$, $P_j\ne 0$ and $\theta_j$ are the frequency, amplitude and phase displacement of the $j$-$th$ frequency component of the $l$-$th$ input, respectively.
As a consequence, the periodical input will result in sustained oscillations present in the measurement vector $\mathbf{y}$. 
The generator/load associated with input $l$ is termed as the \emph{forced oscillation source}, and the measurements at the bus directly connecting to the forced oscillation source are termed as \emph{source measurements}.

In particular, suppose that the frequency $\omega_d$ of an injection component is close to the frequency of a poorly-damped mode, i.e., there exists $j^*\in\{1, 2, \hdots,n\}$,
\begin{equation}
	\omega_d \approx \Im{\lambda_{j^*}}, \quad {\Re{\lambda_{j^*}}}/{\abs{\lambda_{j^*}}}\approx 0.
	\label{eq: resonance_condition}
\end{equation}
In such a case, oscillations with growing amplitude (i.e., resonance) are observed \cite{7041242}. Hence, \eqref{eq: resonance_condition} is adopted as the \emph{resonance condition} in this paper.

In a power system with PMUs, the measurement vector $\mathbf{y}(t)$ is sampled at a frequency of $f_{\text{s}}$ (samples per second). Within a time interval from the FO starting point up to time instant $t$, the time evolution of the measurement vector $\mathbf{y}(t)$ can be discretized by sampling and represented by a matrix called a \emph{measurement matrix} $Y_t = [y_{p,q}^t]$, which we formally define next. Denote by zero the time instant when the forced oscillations start. The following column concatenation defines the measurement matrix $Y_t$ up to time $t$:
\begin{equation}
	Y_t:=
	\begin{bmatrix}
		\mathbf{y}(0),& \mathbf{y}(1/f_{\text{s}}),& \hdots & \mathbf{y}(\left \lfloor{tf_{\text{s}}}\right \rfloor/f_{\text{s}})
	\end{bmatrix},
	\label{eq:measurement_Y}
\end{equation}
where $\left \lfloor{\cdot}\right \rfloor$ denotes the floor operation. The $i$-$th$ column of the measurement matrix $Y_t$ in \eqref{eq:measurement_Y} suggests the ``snapshot'' of all synchrophasor measurements over system at the time $(i-1)/f_{\text{s}}$. The $k$-$th$ row of $Y_t$ denotes the time evolution of the $k$-$th$ measurement deviation in the output vector of the $k$-th PMU. Due to the fact that the output vector may contain multiple types of measurements (e.g., voltage magnitudes, frequencies, etc.), a normalization procedure is introduced as follows. Assume that there are $K$ measurement types. Denote by $Y_{t,i}=[y_{p,q}^{t,i}]$ the measurement matrix of measurement type $i$, where $i = \{1,2,\ldots,K\}$. The normalized measurement matrix $Y_{\text{n}t}=[y_{p,q}^{\text{n},t}]$ is defined by
\begin{equation}
	\label{eq:nomalized_measurement}
	Y_{\text{n}t} =
	\begin{bmatrix}
		\frac{Y_{t,1}^{\top}}{\norm{Y_{1t}}_{\infty}}, & \frac{Y_{t,2}^{\top}}{\norm{Y_{2t}}_{\infty}},&\ldots&\frac{Y_{t,K}^{\top}}{\norm{Y_{Kt}}_{\infty}}
	\end{bmatrix}^{\top}.
\end{equation}

The forced oscillation localization problem is equivalent to pinpointing the forced oscillation source using measurement matrix $Y_t$. Due to the complexity of power system dynamics, the precise power system model \eqref{eq:state_space} may not be available to system operators, especially in a real-time operation. Therefore, it is assumed that the only known information for forced oscillation localization is the measurement matrix $Y_t$. In brief, the first-principle model \eqref{eq:state_space} as well as the perturbation model \eqref{eq:input_sig} is introduced mainly for the purpose of defining FO localization problem and theoretically justifying the data-driven method proposed in Section \ref{sec:problem_formulation}.

	\subsection{Main Challenges of Pinpointing the Sources of Forced Oscillation} 
	\label{sub:challenge_of_pinpointing_the_forced_oscillation_sources}
	The topology of the power system represented by \eqref{eq:state_space} can be characterized by an undirected graph $G=(\mathcal{B}, \mathcal{T})$, where vertex set $\mathcal{B}$ comprises all buses in the power system, while edge set $\mathcal{T}$ collects all transmission lines. Suppose that the PMU measurements at bus $i^*\in \mathcal{B}$ are the source measurements. Then bus $j^*$ is said to be in the vicinity of the FO source, if bus $j^*$ is a member of the following vicinity set:
	\begin{equation}
		\mathcal{V} = \{j\in \mathcal{B}|d_G(i^*, j) \le N_0\},
		\label{eq: vincinity_set}
	\end{equation}
	where $d_G(i,j)$ denotes the $i$-$j$ distance, viz., the number of transmission lines (edges) in a shortest path connecting buses (vertices) $i$ and $j$; the threshold $N_0$ is a nonnegative integer. In particular, $\mathcal{V} = \{i^*\}$ for the source measurement at bus $i^*$, if $N_0$ is set to zero.

	Intuitively, it is tempting to assume that the source measurement can be localized by finding the maximal absolute element in the normalized measurement matrix $Y_{\text{n}t}$, i.e., expecting that the most severe oscillation should be manifested in the vicinity of the source. However, a major challenge for pinpointing the FO sources arises from the following (perhaps counter-intuitive) fact: the most severe oscillation does not necessarily manifest near the FO source, in the presence of \emph{resonance phenomena}. Following the same notation as in \eqref{eq:measurement_Y} and \eqref{eq: vincinity_set}, we term a normalized measurement matrix $Y_{\text{n}t}$ as \emph{counter-intuitive case}, if
	\begin{equation}
	\label{eq:counter_intuitive_1}
		p^* \notin \mathcal{V}, 
	\end{equation}
	where $p^*$ can be obtained by finding the row index of the maximal element in the measurement matrix $Y_t$, i.e.,
	\begin{equation}
	\label{eq:counter_intuitive_2}
		[p^*, q^*] = \text{arg}\max_{p,q} y_{p,q}^{\text{n},t}.
	\end{equation}
	It is such counter-intuitive cases that make pinpointing the FO source challenging \cite{7041242}.
	Figure \ref{fig:counter-intuitive} illustrates one such counter-intuitive case, where the source measurement (red) does not correspond to the most severe oscillation. Additional examples of counter-intuitive cases can be found in \cite{8362302}. Although the counter-intuitive cases are much less likely to happen than the intuitive ones (in terms of frequency of occurrence), it is still imperative to design an algorithm to pinpoint the FO source even in the counter-intuitive cases due to the hazardous consequences of the forced oscillations under resonance conditions.
\begin{figure}[h!]
	\centering
	\includegraphics[width = 1.8in, height = 1.2in]{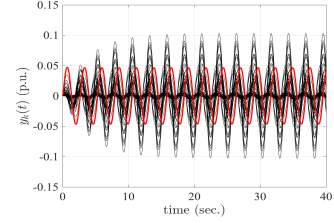}
	\caption{One counter-intuitive case\cite{8362302} from the IEEE 68-bus benchmark system \cite{207380}: the black curves correspond to the non-source measurements; the red curve corresponds to the source measurement.}
	\label{fig:counter-intuitive}
\end{figure}
\section{Problem Formulation and Proposed Methodology} 
\label{sec:problem_formulation}

In this section, we formulate the FO localization problem as a matrix decomposition problem. Besides, we present a FO localization algorithm for real-time operation.

\subsection{Problem Formulation} 
\label{sub:problem_formulation}
	Given a measurement matrix $Y_t$ up to time $t$ with one type of measurement (without loss in generality), the FO source localization is formulated as decomposing the measurement matrix $Y_t$ into a low-rank matrix $L_t$ and a sparse matrix $S_t$:
	\begin{subequations} \label{eq: formulation_1}
	    \begin{align}
	        Y_t = L_t &+ S_t,\\
	        \rank L_t &\le \alpha,\\
	        \norm{S_t}_0 &\le \beta,
	    \end{align}
	 \end{subequations}
	 where the pseudo-norm $\norm{\cdot}_0$ returns the number of non-zero elements of a matrix; the non-negative integer $\alpha$ is the upper bound of the rank of the low-rank matrix $L_t$, and the non-negative integer $\beta$ is the upper bound on the number of non-zero entries in the sparse matrix $S_t$. Given non-negative integers $\alpha$ and $\beta$, it is possible to numerically find $\{L_t,S_t\}$ via \emph{alternating projections} \cite{8362302}. The source measurement index $i^*$ can be tracked by finding the largest absolute value in the sparse matrix $S_t$, viz.,
	 \begin{equation}
	 	 [p^*, q^*]^{\top} = \text{arg}\max_{p,q}{\abs{s_{p,q}^{t}}}.
	 	 \label{eq:i_star}
	 \end{equation}

	 Due to the prior unavailability of the upper bounds $\alpha$ and $\beta$ \cite{8362302}, the matrix decomposition problem shown in \eqref{eq: formulation_1} is reformulated as a instance of \emph{Robust Principal Component Analysis (RPCA)} \cite{candes2011robust}:
	 \begin{equation}
			\min_{S_t} \quad \norm{Y_t-S_t}_{\star} + \xi \norm{S_t}_1,
			\label{eq:formulation_2}
	 \end{equation}
	 where $\norm{\cdot}_{\star}$ and $\norm{\cdot}_1$ denote the nuclear norm and $l_1$ norm, respectively; the tunable parameter $\xi$ regulates the extent of sparsity in $S_t$. The formulation in \eqref{eq:formulation_2} is a convex relaxation of \eqref{eq: formulation_1}. Under some assumptions, the sparse matrix $S_t$ and the low-rank matrix $L_t$ can be disentangled from the measurement matrix $Y_t$ \cite{candes2011robust} by diverse algorithms \cite{lin2010augmented}. The exact Lagrange Multiplier Method (ALM) is used for numerically solving the formulation \eqref{eq:formulation_2}. For a measurement matrix containing multiple measurement types, \eqref{eq:formulation_2} can be modified by replacing $Y_t$ with $Y_{\text{n}t}$.
	
\subsection{FO Localization Algorithm for Real-time Operation} 
\label{sub:real_time_fo_localization_procedure}
Next, we present an FO localization algorithm for real-time operation, using the formulation \eqref{eq:formulation_2}. The starting point of forced oscillations can be determined by the event detector and classifier reported in \cite{6808416,7312503}. Once the starting point of forced oscillations is detected, the forced oscillation source can be pinpointed by Algorithm \ref{alg:FO_Localization}, where $T_0$ and $\xi$ are user-defined parameters.

\begin{algorithm}
	\caption{Real-time FO Localization}\label{alg:FO_Localization}
	\begin{algorithmic}[1]
		\State Update $Y_{T_0}$ by \eqref{eq:measurement_Y};
		\State Obtain $Y_{\text{n}T_0}$ by \eqref{eq:nomalized_measurement};
		\State Find $S_t$ in \eqref{eq:formulation_2} via the exact ALM for chosen $\xi$;
		\State Obtain $p^*$ by \eqref{eq:i_star};
		\State \Return $p^*$ as \label{mark}the source measurement index.
	\end{algorithmic}
\end{algorithm}

\section{Theoretical Interpretation of the RPCA-based Algorithm} 
\label{sec:theoretical_justification_of_the_proposed_methodology}
	This section aims to develop a theoretical connection between the first-principle model in Section \ref{sec:localization_of_forced_oscillation and its challeges} and the data-driven approach presented in Section \ref{sec:problem_formulation}. We start such an investigation by deriving the time-domain solution to PMU measurements in a power system under resonance conditions. Then, the resonance component matrix for the power grid is obtained from the derived solution to PMU measurements. Finally, the efficacy of the proposed method is interpreted by examining the rank of the resonance component matrix.
	\subsection{PMU Measurement Decomposition} 
	\label{sub:pmu_measurement_decompostion}
		For the power system with $r$ inputs and $m$ PMU measurements modeled using \eqref{eq:state_space}, the $k$-th measurement and the $l$-th input can be related by
	\begin{subequations}\label{eq:state_space_SISO}
		\begin{align}
			\mathbf{\dot{x}}(t)&=A\mathbf{x}(t)+\mathbf{b}_lu_l(t)\\ 
			y_k(t)&=\mathbf{c}_k\mathbf{x}(t),
		\end{align}
\end{subequations}
where column vector $\mathbf{b}_l\in \mathbb{R}^n$ is the $l$-th column of matrix $B$ in \eqref{eq:state_space}, and row vector $\mathbf{c}_k \in \mathbb{R}^n$ is the $k$-th row of matrix $C$. Let $\mathbf{x} = M\mathbf{z}$, where $\mathbf{z}$ denotes the transformed state vector and matrix $M$ is chosen such that the similarity transformation of $A$ is diagonal, then
\begin{subequations}\label{eq:decoupled_state_space}
		\begin{align}
			\mathbf{\dot{z}}(t)&=\mathbf{\Lambda}\mathbf{z}(t)+M^{-1}\mathbf{b}_lu_l(t)\\ 
			y_k(t)&=\mathbf{c}_kM\mathbf{z}(t),
		\end{align}
\end{subequations}
where $\mathbf{\Lambda} = \text{diag}(\lambda_1, \lambda_2, \ldots, \lambda_n)=M^{-1} A M$. Denote by column vector $\mathbf{r}_i\in \mathbb{C}^{n}$ and row vector $\mathbf{l}_i\in \mathbb{C}^{n} $ the right and left eigenvectors associated with the eigenvalue $\lambda_i$, respectively.
Accordingly, the transformation matrices $M$ and $M^{-1}$ can be written as $[\mathbf{r}_1,\mathbf{r}_2, \ldots, \mathbf{r}_n]$ and $[\mathbf{l}_1^{\top},\mathbf{l}_2^{\top}, \ldots, \mathbf{l}_n^{\top}]^{\top}$, respectively. The transfer function in the Laplace domain from $l$-th input to $k$-th output is
\begin{equation}
	H(s) = \mathbf{c}_kM(sI- \Lambda)^{-1}M^{-1} \mathbf{b}_l
		= \sum_{i=1}^{n}\frac{\mathbf{c}_k\mathbf{r}_i\mathbf{l}_i \mathbf{b}_l}{s-\lambda_i}.
		\label{eq:transfer_function}
\end{equation}

For simplicity, assume that the periodical injection $u_l$ only contains one component with frequency $\omega_d$ and amplitude $P_d$, namely, $F=1$, $\omega_1 = \omega_d$ and $P_1 = P_d$ in \eqref{eq:input_sig}. Furthermore, we assume that, before $t=0^-$, the system is in steady state, viz., $\mathbf{x}(0^-) = \mathbf{0}$. Let sets $\mathcal{N}$ and $\mathcal{M}'$ collect the indexes of real eigenvalues and the indexes of complex eigenvalues with positive imaginary parts, respectively, viz.,
\begin{equation}
	\mathcal{N} = \{i\in \mathbb{Z}^+|\lambda_i \in \mathbb{R}\}; \quad \mathcal{M}' = \{i\in \mathbb{Z}^+|\Im(\lambda_i)>0\}.
\end{equation}

Then the Laplace transform for PMU measurement $y_k$ is
\begin{equation}
	\begin{aligned}
		&Y_k(s)=\left(\sum_{i=1}^{n}\frac{\mathbf{c}_k\mathbf{r}_i\mathbf{l}_i \mathbf{b}_l}{s-\lambda_i}\right)\frac{P_d \omega_d}{s^2+ \omega_d^2}\\
		&= \left[
			\sum_{i\in \mathcal{N}}\frac{\mathbf{c}_k\mathbf{r}_i\mathbf{l}_i \mathbf{b}_l}{s-\lambda_i}
			+
			\sum_{i\in \mathcal{M'}}\left(
			\frac{\mathbf{c}_k\mathbf{r}_i\mathbf{l}_i \mathbf{b}_l}{s-\lambda_i}
			+
			\frac{\mathbf{c}_k\bar{\mathbf{r}}_i\bar{\mathbf{l}}_i \mathbf{b}_l}{s-\bar{\lambda}_i}
			\right)
		\right]
		\frac{P_d \omega_d}{s^2+ \omega_d^2}
	\end{aligned}
\end{equation}
where $\bar{(\cdot)}$ denotes complex conjugation.

Next, we analyze the components resulting from the real eigenvalues and the components resulting from the complex eigenvalues, individually.
\subsubsection{Components resulting from real eigenvalues} 
	\label{ssub:components_resulting_from_real_eigenvalues}
	In the Laplace domain, the component resulting from a real eigenvalue $\lambda_i$ is
	\begin{equation}
		Y_{k,i}^{\text{D}}(s) = \frac{\mathbf{c}_k\mathbf{r}_i\mathbf{l}_i \mathbf{b}_l}{s-\lambda_i}\frac{P_d \omega_d}{s^2+ \omega_d^2}.
	\end{equation}
	The inverse Laplace transform of $Y_{k,i}^{\text{D}}$(s) is
	\begin{equation}
		 y_{k,i}^{\text{D}}(t) = \frac{\mathbf{c}_k\mathbf{r}_i\mathbf{l}_i \mathbf{b}_l P_d \omega_d}{\lambda_i^2 + \omega_d^2}e^{\lambda_it} + \frac{\mathbf{c}_k\mathbf{r}_i\mathbf{l}_i \mathbf{b}_l P_d}{\sqrt{\lambda_i^2 + \omega_d^2}}\sin(\omega_dt + \phi_{i,l})
	\end{equation}
	where $\phi_{i,l} = \angle\left( \sqrt{\lambda_i^2+ \omega_l^2} + j\lambda_i\right)$, and $\angle(\cdot)$ denotes the angle of a complex number.

\subsubsection{Components resulting from complex eigenvalues} 
\label{ssub:components_resulting_from_the_}
In the Laplace domain, the component resulting from a complex eigenvalue $\lambda_i = - \sigma_i + j \omega_i$ is
\begin{equation}
	Y^{\text{B}}_{k,i}(s)= \left(
			\frac{\mathbf{c}_k\mathbf{r}_i\mathbf{l}_i \mathbf{b}_l}{s-\lambda_i}
			+
			\frac{\mathbf{c}_k\bar{\mathbf{r}}_i\bar{\mathbf{l}}_i \mathbf{b}_l}{s-\bar{\lambda}_i}
			\right)\frac{P_d \omega_d}{s^2+ \omega_d^2}.
\end{equation}
The inverse Laplace transform of $Y^{\text{B}}_{k,i}$(s) is 
\begin{equation}
\begin{aligned}
	&y^{\text{B}}_{k,i}(t)=\\
	&\frac{2 P_d \omega_d \abs{\mathbf{c}_k \mathbf{r}_i \mathbf{l}_i \mathbf{b}_l}}{\sqrt{(\sigma_i^2 + \omega_d^2 - \omega_i^2)^2 + 4 \omega_i^2 \sigma_i^2}} e^{-\sigma_i t} \cos(\omega_i t + \theta_{k,i} -\psi_i)+\\
		&\frac{2 P_d \abs{\mathbf{c}_k \mathbf{r}_i \mathbf{l}_i \mathbf{b}_l}\sqrt{\omega_d^2 \cos^2\theta_{k,i} + ( \sigma_i \cos \theta_{k,i}- \omega_i \sin \theta_{k,i})^2}}{\sqrt{(\sigma_i^2 - \omega_d^2 + \omega_i^2)^2 + 4 \omega_d^2 \sigma_i^2}}\times\\
		&\cos(\omega_d t + \phi_i - \alpha_i),
		\label{eq:beat_time_domain}
\end{aligned}
	\end{equation}
	where $\theta_{k,i} = \angle (\mathbf{c}_k \mathbf{r}_i \mathbf{l}_i \mathbf{b}_l)$; $\psi_i = \angle \left( {\sigma_i^2 + \omega_d^2 - \omega_i^2} - j{2 \sigma_i \omega_i}\right)$; $\phi_i = \angle (\sigma_i^2 - \omega_d^2 + \omega_i^2 - j2 \omega_i \sigma_i)$, and $\alpha_i = \angle [{\omega_d \cos \theta_{k,i}} + j({\sigma_i \cos \theta_{k,i} - \omega_i \sin \theta_{k,i}})]$.

\subsubsection{Resonance component} 
\label{ssub:resonance_components}
Under the resonance condition defined in \eqref{eq: resonance_condition}, the injection frequency $\omega_d$ is in the vicinity of one natural modal frequency $\omega_{j^*}$, and the real part of the natural mode is small. We define a new set $\mathcal{M}\subset \mathcal{M}'$ as $\mathcal{M} = \{i\in \mathbb{Z}^+|\Im(\lambda_i)>0, |\omega_i- \omega_{j^*}|>\kappa\}$,
where $\kappa$ is a small and nonnegative real number. 

For $i \notin \mathcal{M}\cup \mathcal{N}$, the eigenvalue $\lambda_i = -\sigma_i + j \omega_i$ satisfies $
	\omega_i \approx \omega_{j^*} \approx \omega_d$ and $\sigma_i \approx 0$. Then $\psi_i \approx - \frac{\pi}{2}$, $\phi_i \approx- \frac{\pi}{2}$, and $\alpha_i \approx -\theta_{k,i}$.
Therefore, equation \eqref{eq:beat_time_domain} can be simplified as
	\begin{equation}
		y^{\text{B}}_{k,i}(t)\approx y^{\text{R}}_{k,i}(t)= \frac{P_d \abs{\mathbf{c}_k \mathbf{r}_i \mathbf{l}_i \mathbf{b}_l}}{\sigma_i}(1- e^{-\sigma_i t})\sin(\omega_d t + \theta_{k,i})
		\label{eq:resonance_time_domain}
	\end{equation}
	for $i \notin \mathcal{M}\cup \mathcal{N}$. In this paper, $y^{\text{R}}_{k,i}$ in \eqref{eq:resonance_time_domain} is termed as the \emph{resonance component} in the $k$-th measurement. 

	In summary, a PMU measurement $y_k(t)$ in a power system \eqref{eq:state_space} under resonance conditions can be decomposed into three classes of components, i.e., 
	\begin{equation}
		y_k(t) = \sum_{i \in \mathcal{N}} y_{k,i}^{\text{D}}(t) + \sum_{i \in \mathcal{M}}y^{\text{B}}_{k,i}(t) + \sum_{i\notin \mathcal{M}\cup \mathcal{N}} y^{\text{R}}_{k,i}(t).
		\label{eq: summary_sig_decomp}
	\end{equation}
\subsection{Observations on the resonance component and the resonance-free component} 
\label{sub:observations_on_the_resonance_component}
	\subsubsection{Severe oscillations arising from resonance component} 
	\label{ssub:severe_oscillations_arising_from_the_resonance_component}
	Figure \ref{fig:resonance_resonance_free}(a) visualizes the resonance component of a PMU measurement in the IEEE 68-bus benchmark system. As it can be observed from Figure \ref{fig:resonance_resonance_free}(a), the upper envelop of the oscillation increases concavely at the initial stage before reaching a steady-stage value (about $0.1$ in this case). The closed-form approximation for such a steady-state value is ${P_d \abs{\mathbf{c}_k \mathbf{r}_i \mathbf{l}_i \mathbf{b}_l}}/{\sigma_i}$. For a small positive $\sigma_{j^*}$ associated with eigenvalue $\lambda_{j^*}$, the steady-state amplitude of the resonance component may be the dominant one. If a PMU measurement far away from the source measurements is tightly coupled with the eigenvalue $\lambda_{j^*}$, it may manifest the most severe oscillation, thereby confusing system operators with regards to FO source localization. Therefore, the presence of resonance components may cause the counter-intuitive cases defined by \eqref{eq:counter_intuitive_1}, \eqref{eq:counter_intuitive_2}.

	\begin{figure}[h] 
				\centering
				\subfloat[]{\includegraphics[width=1.4in]{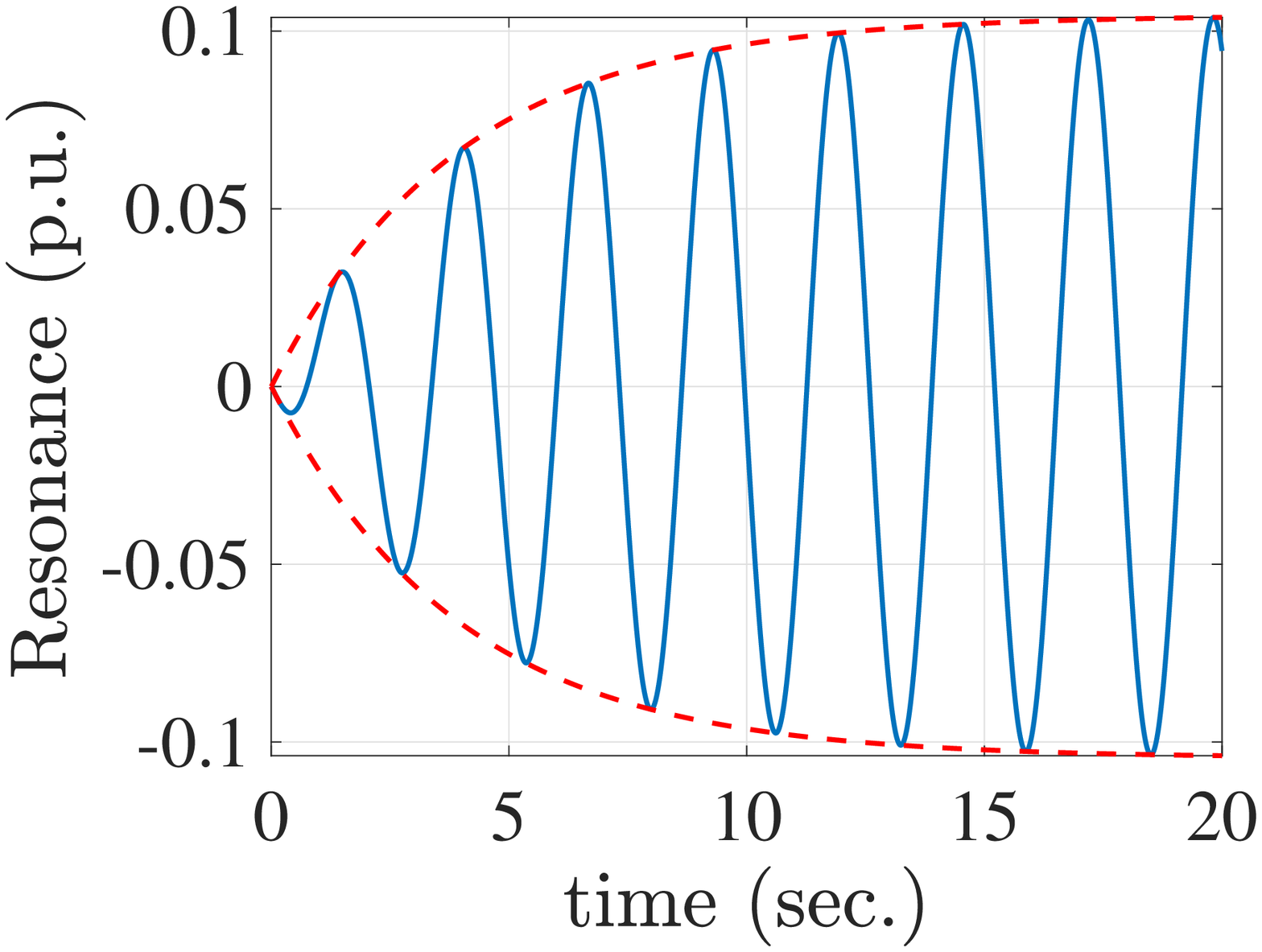}}
				\hfil
				\subfloat[]{\includegraphics[width=1.6in]{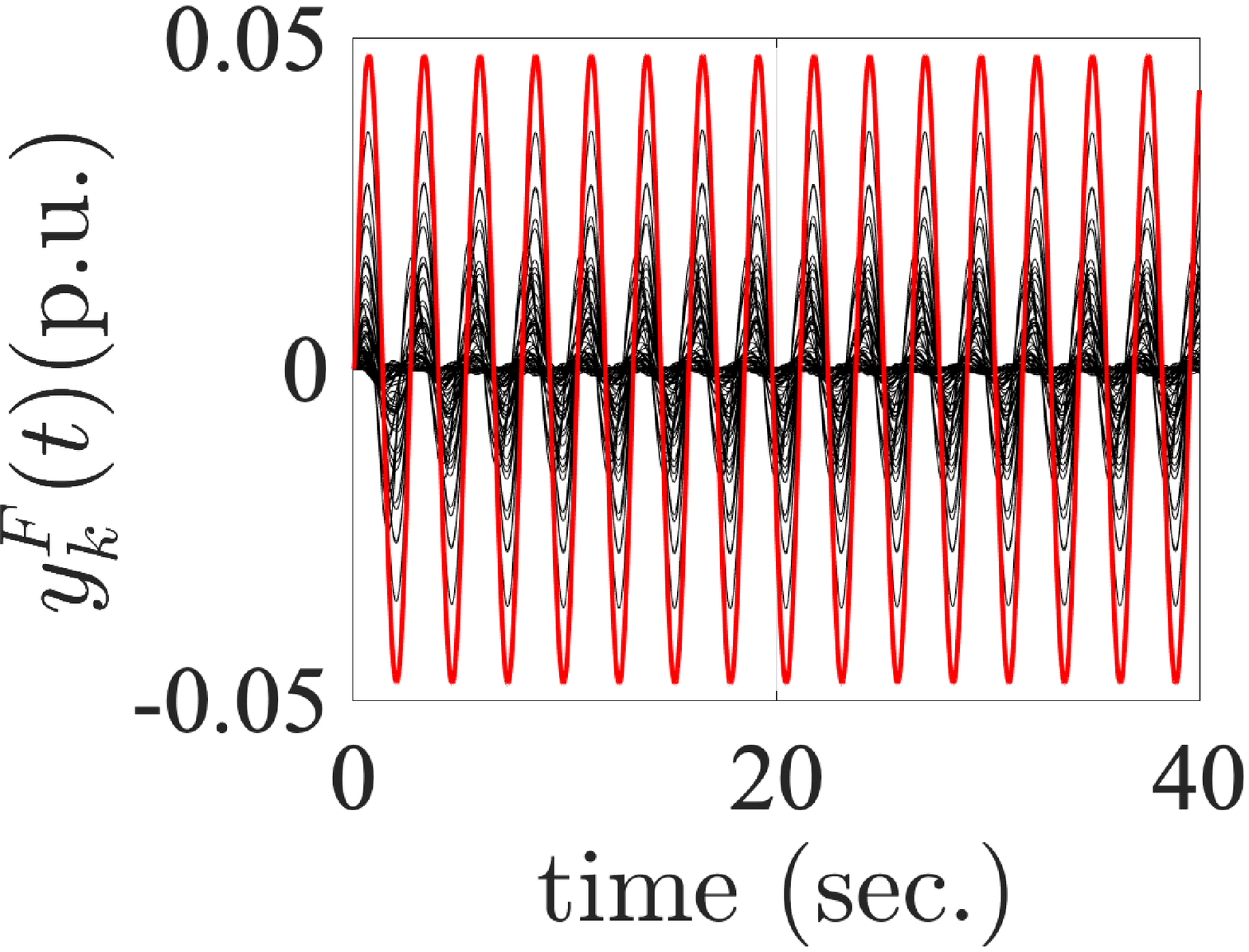}}
				\hfil
				\caption{(a) Visualization of the resonance component of a PMU measurement in the IEEE 68-bus benchmark system based on equation \eqref{eq:resonance_time_domain}: the resonance components of the bus magnitude measurement at Bus 40 (blue curve) and its envelopes (red-dash curves). (b) Resonance-free components of the source measurement (red) and the non-source measurement (black) in the IEEE 68-bus benchmark system.}
				\label{fig:resonance_resonance_free}
			\end{figure}

	\subsubsection{Location information on FO source from the resonance-free component} 
	\label{ssub:descending_trend_of_components_over_space}
	As the resonance components of the set of all PMU measurements mislead system operators with respect to FO localization, we proceed by excluding the resonance component from \eqref{eq: summary_sig_decomp}, and checking whether if the remaining components exhibit any spatial information concerning the FO source. The superposition of the remaining components is termed as \emph{resonance-free}. In specific, for a power system with known physical model \eqref{eq:state_space}, the resonance-free component $y_k^{\text{F}}$ in the $k$-th PMU measurement time series can be obtained by:
	\begin{equation}
		y^{\text{F}}_k(t) = \sum_{i \in \mathcal{N}} y_{k,i}^{\text{D}}(t) + \sum_{i \in \mathcal{M}}y^{\text{B}}_{k,i}(t).
	\end{equation}
	The visualization of the resonance-free component for all PMU measurements in the IEEE 68-bus system is shown in Figure \ref{fig:resonance_resonance_free}(b) under a certain FO setting\footnote{A sinusoidal waveform with amplitude $0.05$ p.u. and frequency $0.38$ Hz is injected into the IEEE 68-bus system via the voltage setpoint of generator 13. The information of the test system is elaborated in Section \ref{sec:case_study}.}. Under the same FO setting, Figure \ref{fig:counter-intuitive} visualizes all PMU measurements $y_k(t)$ in \eqref{eq: summary_sig_decomp}. In Figure \ref{fig:resonance_resonance_free}(b), while the complete measurements $y_k(t)$ are counter-intuitive, the resonance-free components $y_k^{\text{F}}(t)$ convey the location information on the FO source--the resonance-free component of the source measurement exhibits the largest oscillation.


\subsection{Low-rank Nature of Resonance component Matrix} 
\label{sub:low Rankness of Resonance Component}
	The physical interpretation of the efficacy of the RPCA-based algorithm is illustrated by examining the rank of the matrix containing all resonance components for all measurements, which we call the \emph{resonance component matrix} formally defined next. Similar to \eqref{eq:measurement_Y}, the resonance component $y_k^{\text{R}}(t)$ in the $k$-th measurement can be discretized into a row vector $\mathbf{y}_{k,t}^{\text{R}}$:
	\begin{equation} \label{eq:discretized_resonace_vector}
		\mathbf{y}_{k,t}^{\text{R}}:= 
		\begin{bmatrix}
		y_k^{\text{R}}(0),& y_k^{\text{R}}(1/f_{\text{s}}), & \hdots & y_k^{\text{R}}(\left \lfloor{tf_{\text{s}}}\right \rfloor/f_{\text{s}})
	\end{bmatrix}.
	\end{equation}
	Then, the resonance component matrix $Y_t^{\text{R}}$ can be defined as a row concatenation as follows:
	\begin{equation}
	Y_t^{\text{R}}:=
	\begin{bmatrix}
		\left(\mathbf{y}_{1,t}^{\text{R}}\right)^{\top},& \left(\mathbf{y}_{2,t}^{\text{R}}\right)^{\top},& \hdots & \left(\mathbf{y}_{m,t}^{\text{R}}\right)^{\top}
	\end{bmatrix}^{\top}.
	\label{eq:resonance_Y}
\end{equation}

\begin{theorem}
	For the linear time-invariant dynamical system \eqref{eq:state_space}, the rank of the resonance component matrix $Y_t^{\text{R}}$ defined in \eqref{eq:resonance_Y} is at most $2$.
\end{theorem}

\begin{proof}
Based on \eqref{eq:resonance_time_domain}, define $E_k := {P_d \abs{\mathbf{c}_k \mathbf{r}_i \mathbf{l}_i \mathbf{b}_l}}/{\sigma_i}$.
Then
	\begin{equation*}
	\begin{aligned}
		y_k^{\text{R}}(t) = &(1- e^{-\sigma_i t})\sin(\omega_d t)E_k \cos(\theta_{k,i})+ \\& (1- e^{-\sigma_i t})\cos(\omega_d t)E_k \sin(\theta_{k,i}).
	\end{aligned}
\end{equation*}
We further define functions $f_1(t)$, $f_2(t)$ and variables $g_1(k)$, $g_2(k)$ as follows: $f_1(t):=(1- e^{-\sigma_i t})\sin(\omega_d t)$; $f_2(t):=(1- e^{-\sigma_i t})\cos(\omega_d t)$; $g_1(k) :=E_k \cos(\theta_{k,i})$; and $g_2(k) := E_k \sin(\theta_{k,i})$.
Then, $y_k^{\text{R}}(t)$ can be represented by $y_k^{\text{R}}(t) = f_1(t)g_1(k) + f_2(t)g_2(k)$.

The resonance component matrix $Y_t^{\text{R}}$ up to time $t$ can be factorized as follows:

\begin{equation} \label{eq: factorization}
		Y_t^{\text{R}} = \begin{bmatrix}
			g_1(1) & g_2(1)\\
			g_1(2) & g_2(2)\\
			\vdots & \vdots\\
			g_1(m) & g_2(m)
		\end{bmatrix}
		\begin{bmatrix}
			f_1(0)&f_1(\frac{1}{f_{\text{s}}}) & \ldots & f_1(\frac{\left \lfloor{tf_{\text{s}}}\right \rfloor}{f_{\text{s}}})\\
			f_2(0)&f_2(\frac{1}{f_{\text{s}}}) & \ldots & f_2(\frac{\left \lfloor{tf_{\text{s}}}\right \rfloor}{f_{\text{s}}})\\
		\end{bmatrix}.\\
\end{equation}

Denote by vectors $\mathbf{g}_1$ and $\mathbf{g}_2$ the first and second columns of the first matrix in the right hand side (RHS) of \eqref{eq: factorization}, respectively; and by vectors $\mathbf{f}_1$ and $\mathbf{f}_2$ the first and second rows of the second matrix in the RHS of \eqref{eq: factorization}. Then \eqref{eq: factorization} turns to be
\begin{equation} \label{eq: compact_form_factorization}
		Y_t^{\text{R}}=\begin{bmatrix}
			\mathbf{g}_1 & \mathbf{g}_2
		\end{bmatrix}
		\begin{bmatrix}
			\mathbf{f}_1 \\ \mathbf{f}_2
		\end{bmatrix}.
\end{equation}
Given \eqref{eq: compact_form_factorization}, it becomes clear that the rank of the resonance component matrix $Y_t^{\text{R}}$ is at most $2$.
\end{proof}

Typically, for a resonance component matrix $Y_t^{\text{R}}$ with $m$ rows and $\left \lfloor{tf_{\text{s}}}\right \rfloor$ columns, owing to $\min(m, \left \lfloor{tf_{\text{s}}}\right \rfloor)\gg2$, the resonance component matrix $Y_t^{\text{R}}$ is a low-rank matrix, which is assumed to be integrated by the low-rank component $L_t$ in equation \eqref{eq: formulation_1}. As discussed in Section \ref{ssub:descending_trend_of_components_over_space}, the source measurement can be tracked by finding the maximal absolute entry of the resonance-free matrix $(Y_t-Y_t^{\text{R}})$. 
According to \eqref{eq:i_star}, the PMU measurement containing the largest absolute entry in the sparse component $S_t$ is considered as the source measurement.
Then, it is reasonable to conjecture that the sparse component $S_t$ in \eqref{eq: formulation_1} captures the part of the resonance-free matrix that preserves the location information of FO source. Therefore, a theoretical connection between the proposed data-driven method in Algorithm \ref{alg:FO_Localization} and the physical model of power systems described in equation \eqref{eq:state_space} can be established. Although forced oscillation phenomena have been extensively studied in physics \cite{wiener1965cybernetics}, the low-rank property, to the best of our knowledge, is first investigated in this paper.

Note that Theorem $1$ offers one possible interpretation of the effectiveness of the proposed algorithm. As this paper focuses on the development of one possible data-driven localization algorithm, future work shall investigate broader category of possible algorithm and their theoretical underpinning.


\section{Case Study} 
\label{sec:case_study}
In this section, we validate the effectiveness of Algorithm \ref{alg:FO_Localization} using data from IEEE 68-bus benchmark system and WECC 179-bus system. We first describe the key information on the test systems, the procedure for obtaining test data, the parameter settings of the proposed algorithm, and the algorithm performance over the obtained test data. Then the impact of different factors on the performance of the localization algorithm is investigated. Finally, we compare the proposed algorithm with the energy-based method reported in \cite{maslennikov2017dissipating}.
As will be seen, the proposed method can pinpoint the FO sources with high accuracy \emph{without any information on system models and grid topology}, even when resonance exists.
\subsection{Performance Evaluation of the Localization Algorithms in Benchmark Systems} 
\label{sub:System_Description}
\subsubsection{IEEE 68-bus Power System Test Case} 
\label{ssub:ieee_68_bus_power_system}
The system parameters of the IEEE 68-bus power system are reported in the Power System Toolbox (PST) \cite{207380} and its topology is shown in Figure \ref{fig:IEEE68-bus}. Let $\mathcal{V}=\{1,2,\ldots,16\}$ consist of the indexes of all $16$ generators in the 68-bus system. Based on the original parameters, the following modifications are made: 1) the power system stabilizers (PSS) at all generators, except the one at Generator $9$, are removed, in order to create more poorly-damped oscillatory modes; 2) for the PSS at Generator $9$, the product of PSS gain and washout time constant is changed to $250$. Based on the modified system, the linearized model of the power system \eqref{eq:state_space} can be obtained using the command ``\texttt{svm\char`_mgen}'' in PST. There are $25$ oscillatory modes whose frequencies range from $0.1$ Hz to $2$ Hz. Denote by $\mathcal{W}=\{\omega_1, \omega_2,\ldots,\omega_{25}\}$ the set consisting all $25$ modal frequencies of interest. The periodical perturbation $u_l$ in \eqref{eq:input_sig} is introduced through the voltage setpoints of generators. The analytical expression of $u_l$ is $0.05 \sin (\omega_dt)$, where $\omega_d \in \mathcal{W}$.
\begin{figure}[h]
	\centering
 	\includegraphics[width = 3in]{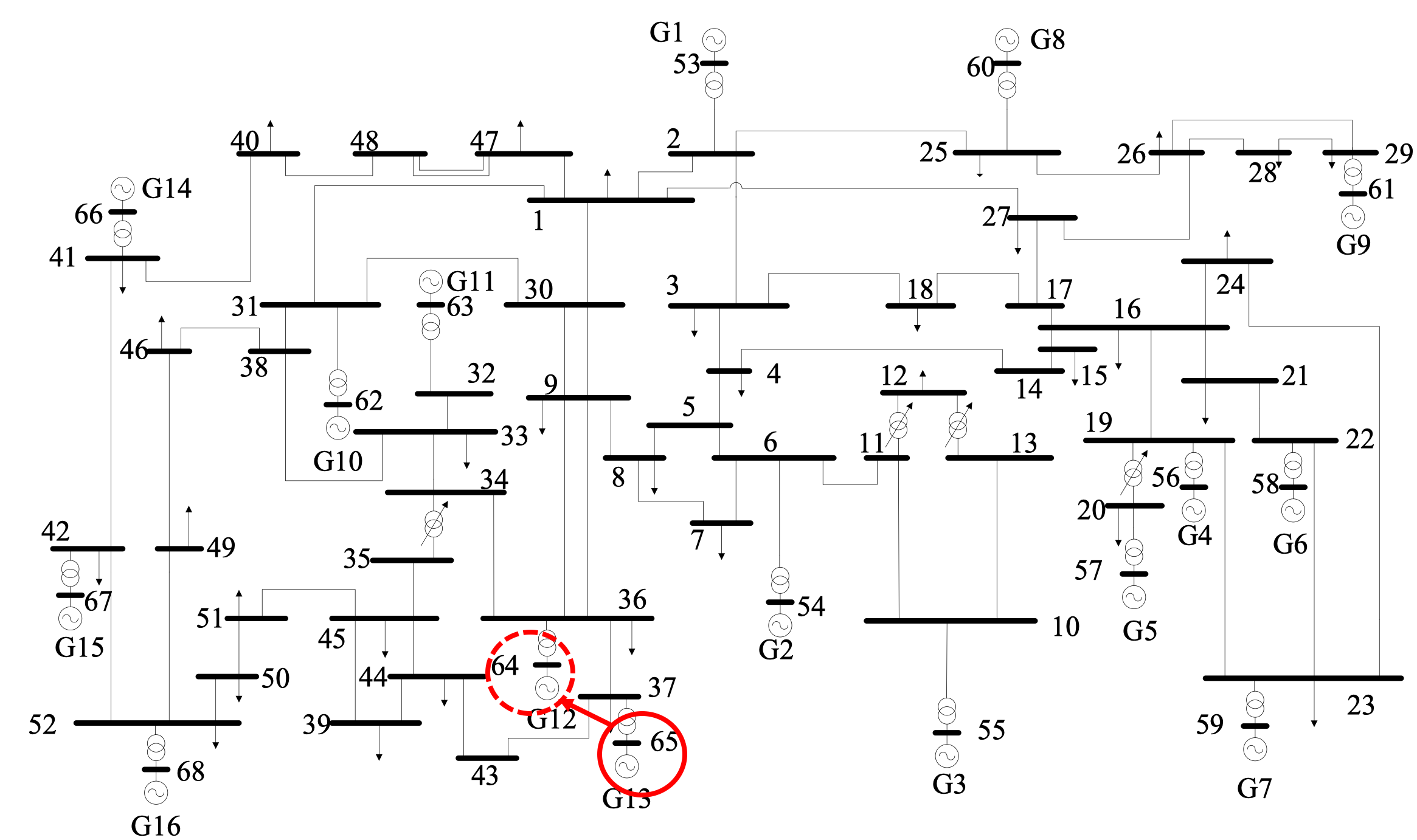}
	\caption{The IEEE 68-bus power system \cite{8362302}: the generator in the solid circle is the actual source generator; the generator in the dash circle is the identified source.}
	\label{fig:IEEE68-bus}
\end{figure}

We create forced oscillations in the 68-bus system according to set $\mathcal{V} \times \mathcal{W}$: for element $(i, \omega_j)\in \mathcal{V} \times \mathcal{W}$, the periodical perturbation $u_l(t)$ with frequency $\omega_j$ is injected into the grid through the voltage setpoint of generator $i$ at time $t=0$. Then, the system response is obtained by conducting a $40$-second simulation. The bus voltage magnitude deviations constitute the output/measurement vector $\mathbf{y}(t)$ in \eqref{eq:state_space}. Finally, the measurement matrix is constructed based on \eqref{eq:measurement_Y}, where the sampling rate $f_{\text{s}}$ is $60$ Hz. By repeating the above procedure for each element in set $\mathcal{V}\times \mathcal{W}$, we obtain $400$ measurement matrices ($|\mathcal{V}\times \mathcal{W}|$). For the $400$ measurement matrices, $44$ measurement matrices satisfy the resonance criteria \eqref{eq:counter_intuitive_1}, \eqref{eq:counter_intuitive_2} with $N_0 = 0$ and they are marked as the counter-intuitive cases which are used for testing the performance of the proposed method.

The tunable parameters $T_0$ and $\xi$ in Algorithm \ref{alg:FO_Localization} are set to $10$ and $0.0408$, respectively. The detailed information on setting $\xi$ can be found in \cite{lin2010augmented}. Measurements of voltage magnitude, phase angle and frequency are used for constituting the measurement matrix. Then, we apply Algorithm \ref{alg:FO_Localization} for the $44$ counter-intuitive cases. Algorithm \ref{alg:FO_Localization} can pinpoint the source measurements in $43$ counter-intuitive cases and, therefore, achieves $97.73\%$ accuracy without any knowledge of system models and grid topology.

Next, we scrutinize the geographic proximity between the identified and actual source measurements in the single failed case. The algorithm outputs that the source measurement is located at Bus $64$ (highlighted with a solid circle in Figure \ref{fig:IEEE68-bus}), when a periodic perturbation with frequency $1.3423$ Hz is injected into the system through the generator directly connecting to Bus $65$ (highlighted with a dash circle in Figure \ref{fig:IEEE68-bus}). As it can be seen in Figure \ref{fig:IEEE68-bus}, the identified and actual source measurements are geographically close. Therefore, even the failed cases from the proposed method can effectively narrow the search space.


\subsubsection{WECC 179-bus System Test Case} 
This subsection leverages the open-source forced oscillation dataset \cite{7741772} to validate the performance of the RPCA-based method. The offered dataset is generated via the WECC 179-bus power system \cite{7741772} whose topology is shown in Figure \ref{fig:WECC179}(a). The procedure of synthesizing the data is reported in \cite{7741772}. The available dataset includes $15$ forced oscillation cases with single oscillation source, which are used to test the proposed method. In each forced oscillation case, the measurements of voltage magnitude, voltage angle and frequency at all generation buses are used to construct the measurement matrix $Y_t$ in \eqref{eq:measurement_Y}, from the $10$-second oscillatory data, i.e., $T_0 = 10$. Then, the $15$ measurement matrices are taken as the input for Algorithm \ref{alg:FO_Localization}, where the tunable parameter $\xi$ is set to $0.0577$ using the same reasoning as in the 68-bus system case.

For the WECC 179-bus system, the proposed method achieved $93.33\%$ accuracy. Next, we present how geographically close the identified FO sources are to the ground truth in the seemingly incorrect case. In Case \texttt{FM-6-2}, a periodic rectangular perturbation is injected into the grid through the governor of the generator at Bus $79$ which is highlighted with a red solid circles in Figure \ref{fig:WECC179}(b). The source measurement identified by the proposed method is at Bus $35$ which is highlighted by a red dash circle. As can be seen in Figure \ref{fig:WECC179}(b), the identified FO source is geographically close to the actual source. Again, even the seemingly wrong result can help system operators substantially narrow down the search space for FO sources.

\begin{figure}[h] 
				\centering
				\subfloat[]{\includegraphics[width=1.4in, height=1.8in]{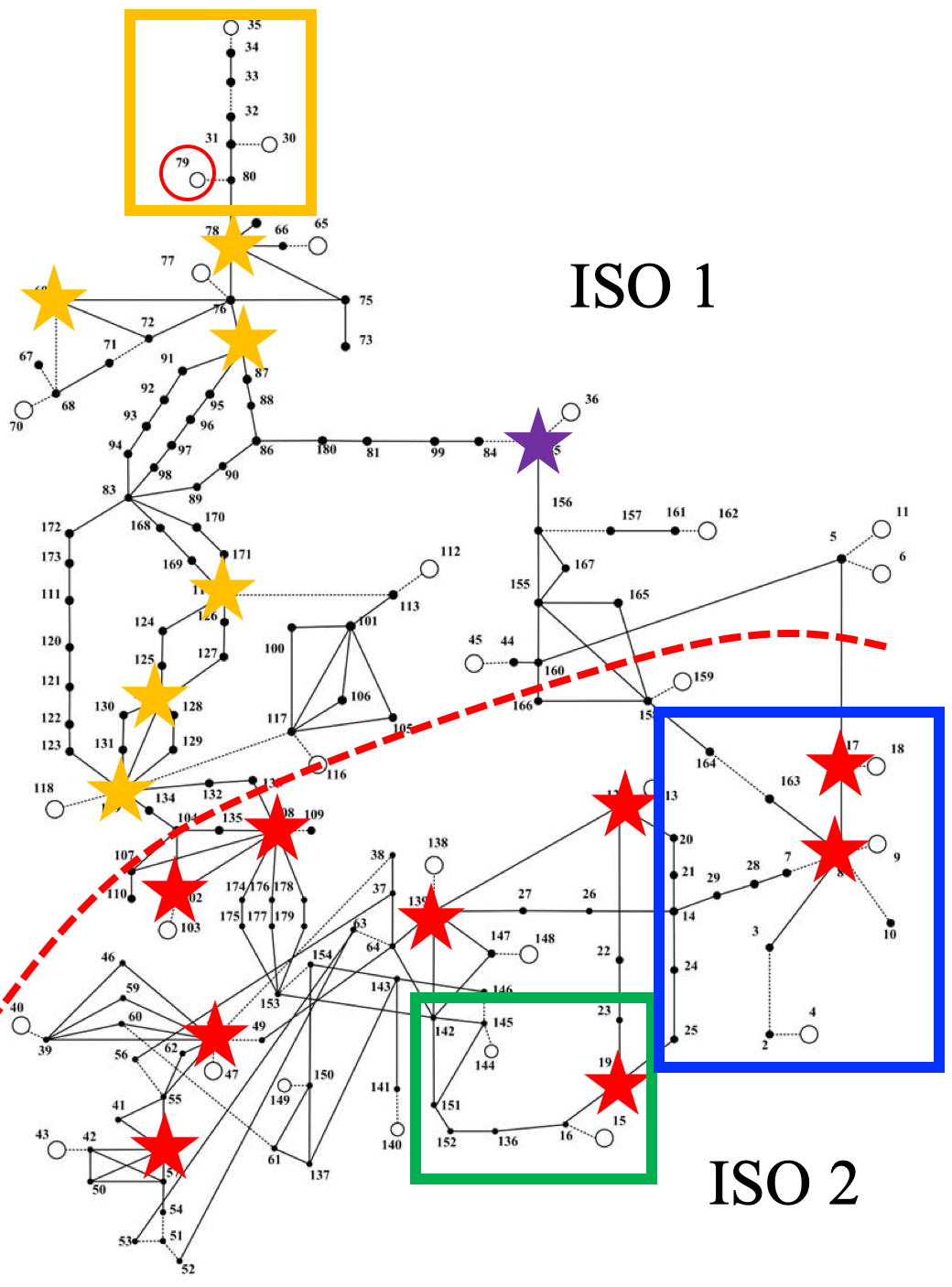}}
				\hfil
				\subfloat[]{\includegraphics[width=1in]{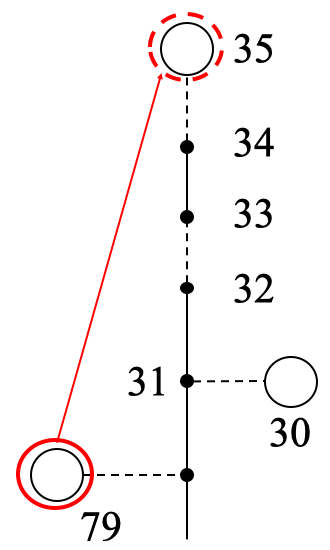}}
				\hfil
				\caption{WECC 179-bus power system \cite{7741772}: (a) complete topology; (b) zoomed-in version of the area in the yellow box in the left figure.}
				\label{fig:WECC179}
			\end{figure}
\subsection{Algorithm Robustness} 
\label{sub:algorithm_robustness}
The subsection focuses on testing the robustness of the proposed algorithm under different factors which include measurement types, noise, and partial coverage of PMUs. The impact of each factor on the algorithm performance will be demonstrated as follows.
\subsubsection{Impact of Measurement Types on Algorithm Performance} 
\label{ssub:impact_of_measurement_types_on_algorithm_performance}
Under all possible combinations of nodal measurements (voltage magnitude $|V|$, voltage angle $\angle V$ and frequency $f$), the localization accuracies of the proposed algorithm in the two benchmark systems are reported in Table \ref{tab:dif_type_meas}. As can be observed in Table \ref{tab:dif_type_meas}, the maximal accuracy is achieved when voltage magnitudes, voltage angles and frequencies are used to constitute the measurement matrix in \eqref{eq:measurement_Y}.
\begin{table}[tb]
	\caption{Impact of Measurement Types on Localization Performance}
	\label{tab:dif_type_meas}
	\centering

	\begin{tabular}{l|c|c|c|c}
	\hline

	\hline
	\textbf{Types}& $|V|$ & $\angle{V}$ & $|V|, \angle{V}$ & $f$\\
	\hline
	$68$-bus System&$84.09 \%$& $50.00\%$& $84.09$\%& $52.27\%$\\
	\hline
	$179$-bus System&$86.67\%$& $33.33\%$& $73.33\%$ & $20.00\%$\\
	\hline
	\textbf{Types}& $|V|, f$& $\angle{V},f$ & $|V|, \angle{V}, f$& N/A\\
	\hline
	$68$-bus System& $93.18\%$ & $59.09\%$ &$\mathbf{97.73\%}$& N/A\\
	\hline
	$179$-bus System&$80.00\%$ & $46.67\%$ &$\mathbf{93.33}\%$& N/A\\

	\hline

	\hline
	\end{tabular}
\end{table}
\subsubsection{Impact of Noise on Algorithm Performance} 
\label{ssub:impact_of_different_noise_levels_on_algorithm_performance}
Table \ref{tab:noise} records the localization accuracy under different levels of noise. We can conclude the proposed algorithm performs well under the cases with signal-to-noise ratio (SNR) less than $30$ dB which is lower than the SNR used for PMU-related tests \cite{7892872}.
\begin{table}[tb]
	\caption{Impact of Noise Levels on Localization Performance}
	\label{tab:noise}
	\centering

	\begin{tabular}{l|c|c|c|c|c}
	\hline

	\hline
	SNR&	$90$dB&	$70$dB&	$50$dB&	$30$dB&	$10$dB\\
	\hline
$68$-Bus&	$97.73$\%&	$97.73$\%&	$97.73$\%&	$97.73$\%&	$56.82$\%\\
\hline
$179$-Bus&	$93.33$\%&	$93.33$\%&	$93.33$\%&	$93.33$\%& $73.33$\%\\

	\hline

	\hline
	\end{tabular}
\end{table}
\subsubsection{Impact of Partial Coverage of Synchrophasors on Algorithm Performance} 
\label{ssub:impact_of_partial_coverage_of_synchrophasors_on_algorithm_performance}
In practice, not all buses are equipped with PMUs. Besides, available PMUs may be installed in buses near oscillation sources, instead of buses to which oscillation sources are directly connected. A test case is designed for testing the performance of the proposed algorithm in the scenario as described above. In this test case, the locations of all available PMUs are marked with stars (regardless colors of the stars) in Figure \ref{fig:WECC179}(a). The test result is listed in Table \ref{tab:partial_observation}. As illustrated in Table \ref{tab:partial_observation},  the proposed method can effectively identify the available PMUs that are close to oscillation sources, even though no PMU is installed in generation buses.
\begin{table*}[tb]
	\caption{Impact of Partial Coverage of Synchrophasor on Algorithm Performance}
	\label{tab:partial_observation}
	\centering

	\begin{tabular}{l|c|c|c|c|c|c|c|c|c|c|c|c|c|c|c}
	\hline

	\hline
	\textbf{Case Name} & F-1 & FM-1 & F-2& F-3&FM-3& F-4-1&F-4-2&F-4-3&F-5-1&F-5-2&F-5-3&F-6-1&F-6-2&F-6-3&FM-6-2 \\
	\hline
	\textbf{Identified Source}& 8&8&78&69&69&69&78&78&78& 78&78&78&78&78&78\\
	\hline
	\textbf{Nearest PMU}&8&8&78/69&78/69&78/69&78/69&78/69&78/69&78/69&78/69&78/69&78/69&78/69&78/69&78/69\\

	\hline

	\hline
	\end{tabular}
\end{table*}

Independent System Operators (ISOs) may also need to know whether FO sources are within their control areas. However, ISOs might not be able to access PMUs near FO sources, limiting the usefulness of the proposed algorithm. For example, assume that there are two ISOs, i.e., ISO $1$ and ISO $2$, in Figure \ref{fig:WECC179}(a), where the red dash line is the boundary between the control areas of the two ISOs. It is possible that FO sources are at the ISO $1$ control area, whereas ISO $2$ only can access the PMUs at the buses marked with red stars. In order to apply the RPCA-based method, ISO $2$ need to access one PMU in the area controlled by ISO $1$, say, the PMU marked with a purple star in Figure \ref{fig:WECC179}(a). In the \texttt{F-2} dataset, the FO source locates at Bus $79$ which is marked with a red circle in Figure \ref{fig:WECC179}(a). With the data collected from PMUs marked with red and purple stars, the proposed algorithm outputs the bus marked with a purple star, indicating that the FO source is outside the control area of ISO $2$.

\subsection{Comparison with Energy-based Localization Method} 
\label{sub:comparison_with_energy_based_localization_method}
This subsection aims to compare the proposed localization approach with the Dissipating Energy Flow (DEF) approach \cite{maslennikov2017dissipating}. We use the \texttt{FM-1} dataset (Bus $4$ is the source measurement) \cite{7741772} for the purpose of comparing DEF method with the proposed algorithm. PMUs are assumed to be installed at all generator buses except ones at Buses $4$ and $15$. Besides, Buses $7$, $15$ and $19$ are also assumed to have PMUs. \emph{Without any information on grid topology}, the RPCA-based method suggests the source measurement is at Bus $7$ which is in the vicinity of the actual source. However, topology errors may cause DEF-based method to incur both false negative and false positive errors, as will be shown in the following two scenarios.

\subsubsection{Scenarios 1} 
\label{ssub:false_negative_errors_resulting_from_topology_errors}
The zoomed-in version of the area within the blue box in Figure \ref{fig:WECC179}(a) is shown in Figure \ref{fig:Type1}, where the left and right figures are actual system topology and the topology reported in a control center, respectively. All available PMUs are marked with yellow stars in Figure \ref{fig:Type1}. Based on these available PMUs, the relative magnitudes and directions of dissipating energy flows are computed according to the \texttt{FM-1} dataset and the method reported in \cite{maslennikov2017dissipating}. With the true topology, the FO source cannot be determined, as the energy flow direction along Branch $8$-$3$ cannot be inferred based on the available PMUs. However, with the topology error shown in Figure \ref{fig:Type1}(b), i.e., it is mistakenly reported that Bus 29 (Bus 17) is connected with Bus 3 (Bus 9), it can be inferred that the an energy flow with relative magnitude of $0.4874$ is injected into the Bus $4$, indicating that Bus $4$ is \emph{not} the source measurement. Such a conclusion contradicts with the ground truth. Therefore, with such a topology error, the dissipating energy flow method leads to a false negative error.
			\begin{figure}[h] 
				\centering
				\subfloat[]{\includegraphics[width=1.3in]{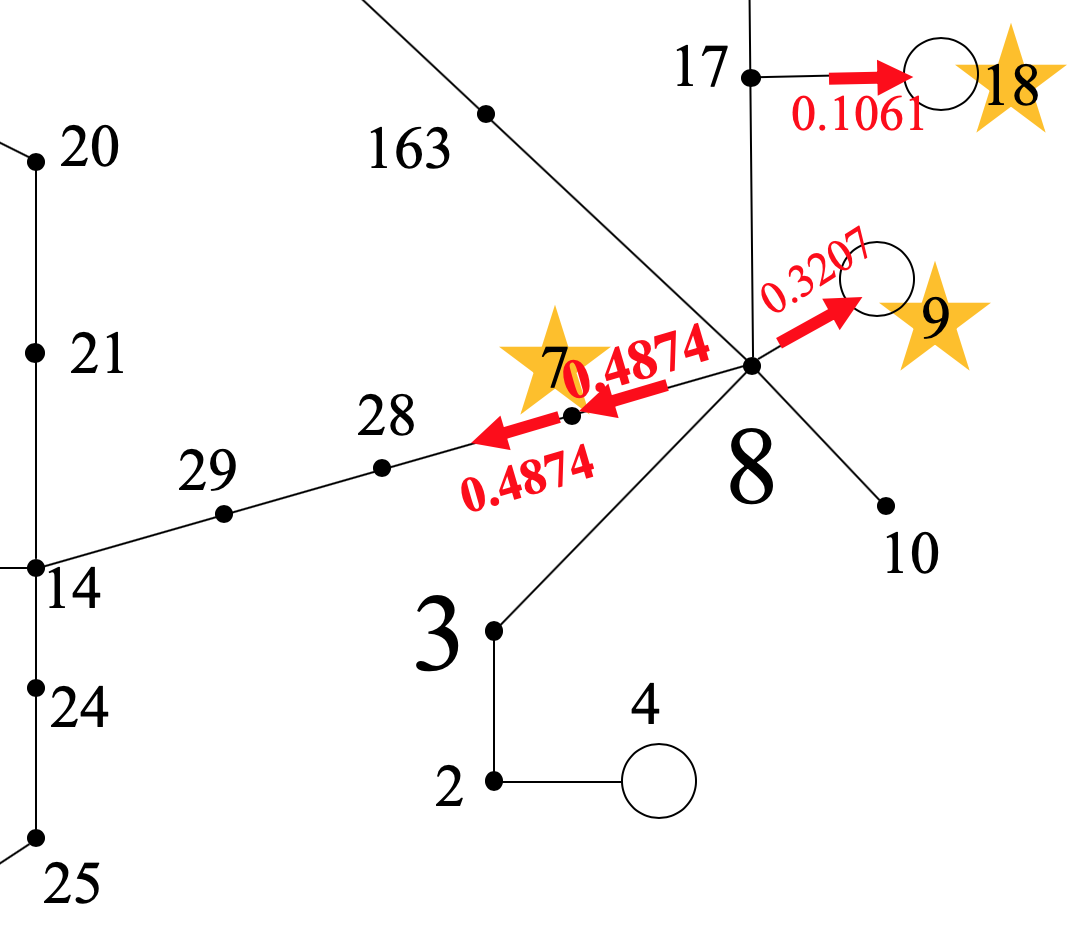}}
				\hfil
				\subfloat[]{\includegraphics[width=1.3in]{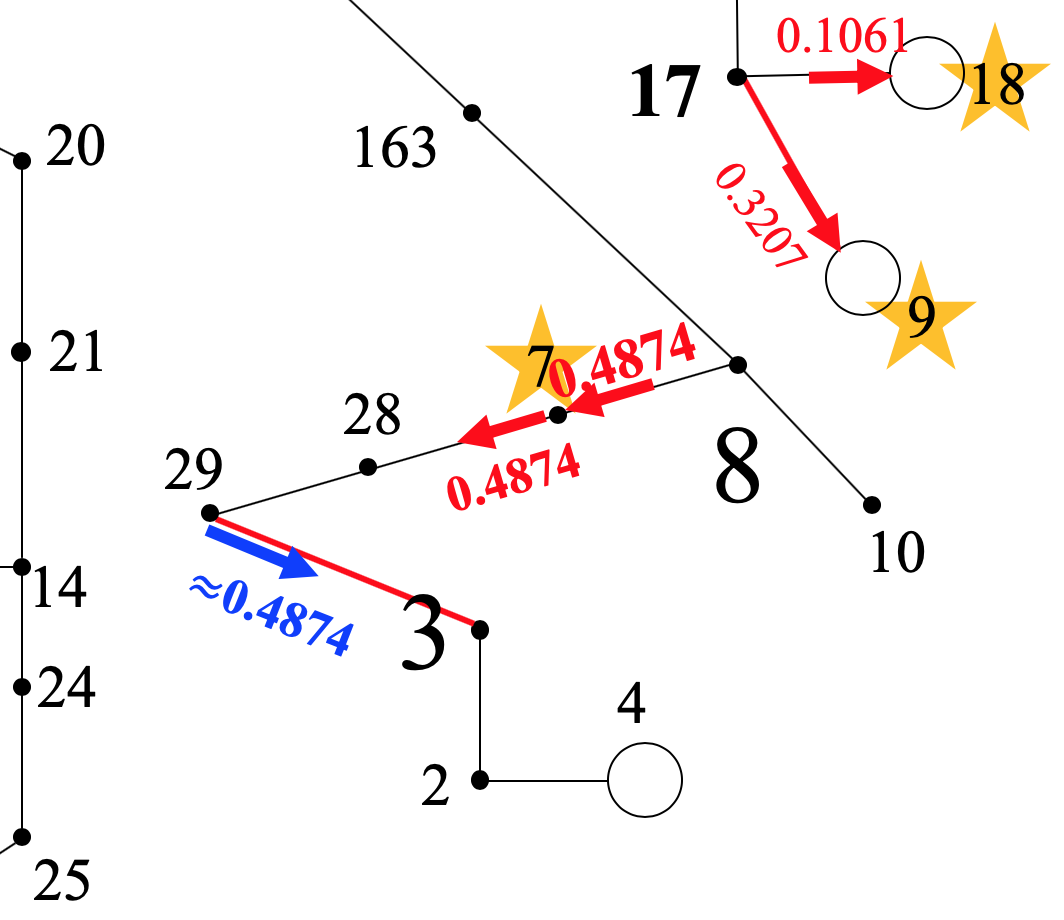}}
				\hfil
				\caption{Zoomed-in version of the area in the blue box at Figure \ref{fig:WECC179} (a): actual topology (left); topology reported in a control center (right). Relative magnitudes and direction of energy flows are labeled with red numbers and arrows, respectively.}
				\label{fig:Type1}
			\end{figure}

\subsubsection{Scenario 2} 
\label{ssub:false_positive_error_resulting_from_topology_errors}
Similar to Scenario 1, topology errors exist within the area highlighted by a green box in Figure \ref{fig:WECC179}(a), whose zoomed-in version is shown in Figure \ref{fig:Type2}. As shown in Figure \ref{fig:Type2}(a), it can be inferred that an energy flow with relative magnitude of $0.171$ injects into Bus $15$ with the information of actual topology and available PMUs, indicating Bus $15$ is not a source. However, with the reported system topology, the generator at Bus $15$ injects to the rest of grid an energy flow with magnitude of $0.0576$, suggesting the source measurement is at Bus $15$. Again, such a conclusion contradicts with the ground truth and, hence, incurs a false positive error.

\begin{figure}[h] 
				\centering
				\subfloat[]{\includegraphics[width=1.2in]{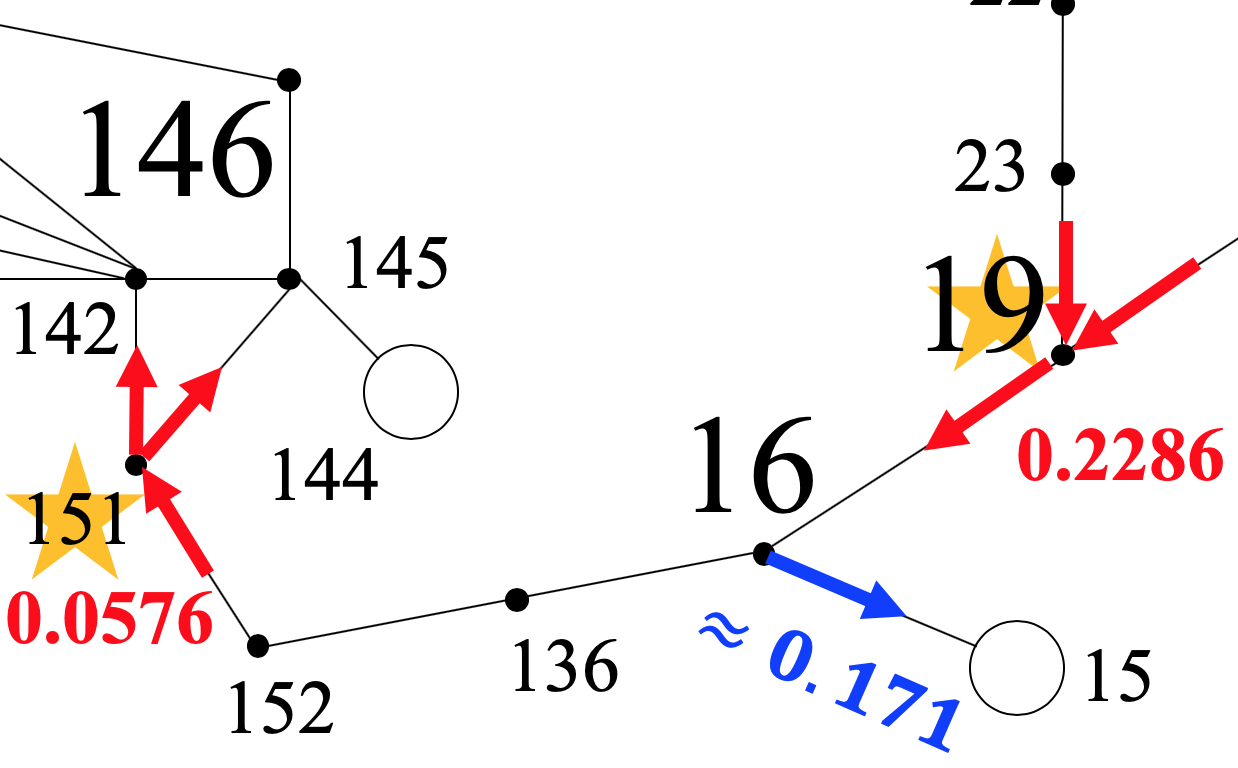}}
				\hfil
				\subfloat[]{\includegraphics[width=1.2in]{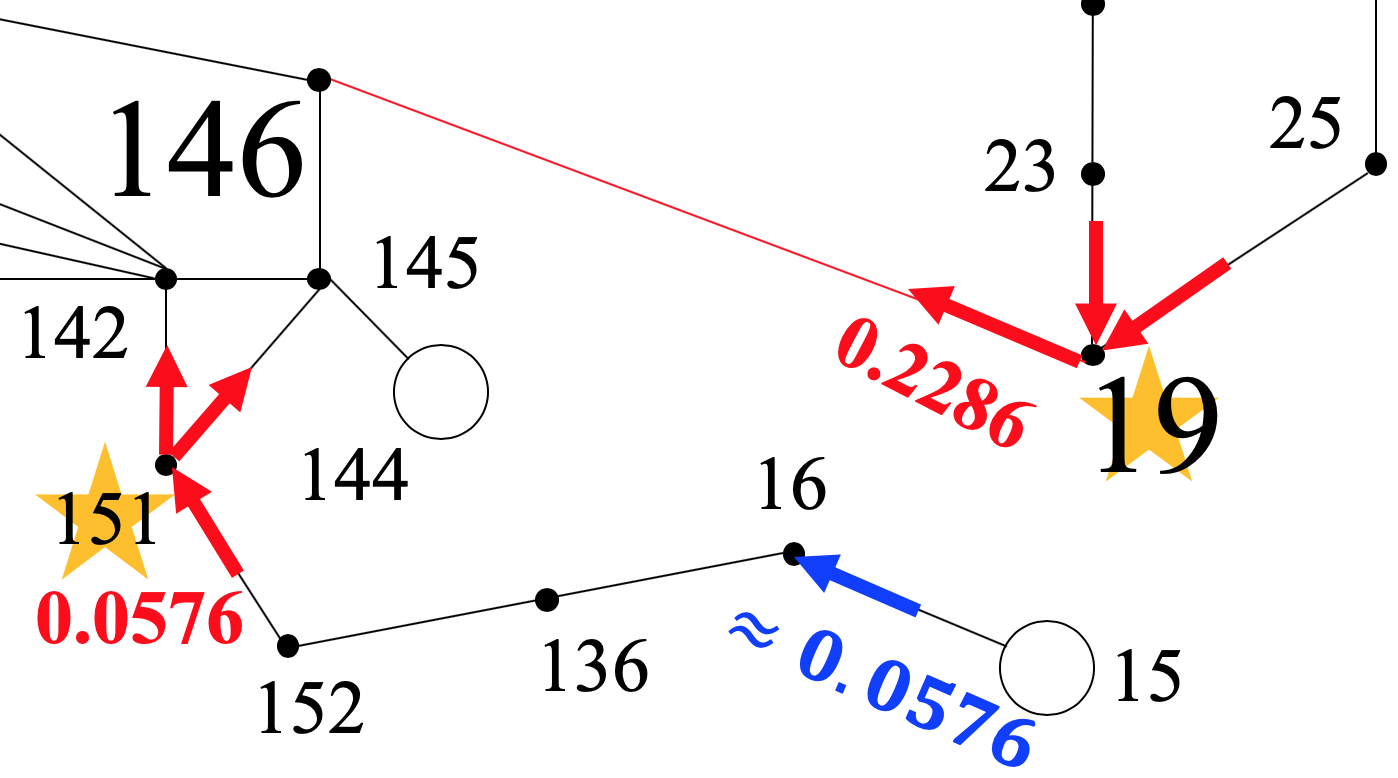}}
				\hfil
				\caption{Zoomed-in version of the area in the green box at Figure \ref{fig:WECC179} (a): actual topology (left); topology reported in a control center (right).}
				\label{fig:Type2}
			\end{figure}


\section{Conclusions} 
\label{sec:conclusion}
In this paper, a purely data-driven but physically interpretable method is proposed in order to locate forced oscillation sources in power systems. The localization problem is formulated as an instance of matrix decomposition, i.e., how to decompose the high-dimensional synchrophasor data into a low-rank matrix and a sparse matrix, which can be done using Robust Principal Component Analysis. Based on this problem formulation, a localization algorithm for real-time operation is presented. The proposed algorithm does not require any information on system models nor grid topology, thus enabling an efficient and easily deployable solution for real-time operation. In addition, a possible theoretical interpretation of the efficacy of the algorithm is provided based on physical model-based analysis, highlighting the fact that the rank of the resonance component matrix is at most 2. Future work will explore a broader set of algorithms and their theoretical performance analysis for large-scale realistic power systems.
\bibliographystyle{IEEEtran}
\bibliography{J3ref}

		




\ifCLASSOPTIONcaptionsoff
  \newpage
\fi






\end{document}